\documentclass[aps,pra,twocolumn,reprint,groupedaddress]{revtex4-1}
\usepackage{graphicx,color}
\usepackage{amsmath}
\usepackage{amssymb}
\usepackage{amsfonts}
\usepackage{bm}
\usepackage[T1]{fontenc}
\usepackage{amsfonts}
\usepackage{amsthm}
\usepackage{graphicx}
\usepackage{color}
\usepackage{float}
\usepackage{nicefrac}
\def\one{{\mathchoice {\mathrm{1\mskip-4mu l}} {\mathrm{1\mskip-4mu l}}
{\mathrm{1\mskip-4.5mu l}} {\mathrm{1\mskip-5mu l}}}}
\renewcommand{\H}{\mathcal{H}}

\newcommand{\B}{\mathcal{B}}
\newcommand{\C}{\mathbb{C}}
\newcommand{\I}{\mathcal{I}}
\newcommand{\R}{\mathcal{R}}
\renewcommand{\O}{\mathbb{O}}

\newcommand{\ot}{\otimes}
\newcommand{\dg}{\dagger}
\newcommand{\tr}{\mathrm{Tr}}
\newcommand{\ket}[1]{|#1\rangle}
\newcommand{\bra}[1]{\langle#1|}
\newcommand{\proj}[2]{|#1\rangle\langle#2|}
\newcommand{\inner}[2]{\langle#1|#2\rangle}
\newcommand{\m}[1]{\mathbb{M}_{#1}}
\newcommand{\ew}[1]{\mathcal{W}_{#1}}

\newtheorem{thm}{Theorem}
\newtheorem{lem}[thm]{Lemma}

\newtheorem{cor}{Corollary}
\newtheorem*{exmp}{Example}
\begin{document}

\title{Recurrent construction of optimal entanglement witnesses\\ for $2N$ qubit systems} 


\author{Justyna P. Zwolak}
\email[]{j.p.zwolak@gmail.com}
\affiliation{Department of Physics, Oregon State University,
\\
Corvallis, OR 97331, USA}
\author{Dariusz Chru\'{s}ci\'{n}ski}
\affiliation{Institute of Physics, Nicolaus Copernicus University,
\\
Grudzi\k{a}dzka 5/7, 87--100 Toru\'{n}, Poland}

\date{\today}

\begin{abstract}
We provide a  recurrent construction of entanglement witnesses for a bipartite systems living in a Hilbert space corresponding to $2N$ qubits ($N$ qubits in each subsystem). Our construction provides a new method of generalization of the Robertson map that naturally meshes with $2N$ qubit systems, i.e., its structure respects the $2^{2N}$ growth of the state space. We prove that for $N>1$ these witnesses are indecomposable and optimal. As a byproduct we provide a new family of PPT (Positive Partial Transpose) entangled states. 
\end{abstract}

\pacs{xxx}


\maketitle


\section{Introduction}

Entanglement witnesses (EW) provide universal tools for analyzing and detecting quantum entanglement \cite{Horodecki09-QE,Guhne09-ED}. Let us recall that a Hermitian operator $\ew{}$ defined on a tensor product $\H=\H_A \ot \H_B$ is called  an EW iff $\bra{\psi_A \ot \phi_B}\ew{}\ket{\psi_A \ot \phi_B } \geq 0$ and $\ew{}$ possesses at least one negative eigenvalue. It turns out that a state $\rho$ in $\H$ is entangled if and only if it is detected by some EW \cite{Horodecki96-SEP}, that is, iff there exists an EW $\ew{}$ such that $\mbox{Tr}(\ew{}\rho)<0$.  In recent years there was a considerable effort in
constructing and analyzing the structure of EWs (see e.g. \cite{Terhal02-DQE,Lewenstein00-OEW,Kraus02-CDS,Bruss02-CE,Toth05-DME,Bertlmann02-GPE,Breuer06-OED,Hall06-IOP,Chruscinski08-HEW,Pytel09-EW1,Pytel10-EW2,Pytel11-EW3,Zwolak13-EW4,Ha04-CES,Ha05-EES,Ha11-IEW,Spengler12-DMB}). However, the general construction of an EW is not known. Let us recall that an entanglement witness $\ew{}$ is decomposable if
\begin{equation}\label{DEC}
  \ew{} = A + B^\Gamma \ ,
\end{equation}
where $A,B \geq 0$ and $B^\Gamma$ denotes a partial transposition of $B$.  EWs that can not be represented as (\ref{DEC}) are called indecomposable. Indecomposable EWs  are necessary to detect PPT entangled states (a state $\rho$ is PPT if $\rho^\Gamma \geq 0$). If $\rho$ is PPT, $\ew{}$ is an EW and ${\rm Tr}(\ew{}\rho) < 0$, then $\rho$ is entangled and $\ew{}$ is necessarily indecomposable. The optimal EW is defined as follows: if $\ew{1}$ and $\ew{2}$ are two entanglement witnesses then following Ref. \cite{Lewenstein00-OEW} we call $\ew{1}$ finer than $\ew{2}$ if $D_{\ew{1}} \supseteq  D_{\ew{2}}$, where
\begin{equation*}\label{}
  D_{\ew{}} = \{ \, \rho\, |\, {\rm Tr}(\rho \ew{}) < 0 \, \}\ 
\end{equation*}
denotes the set of all entangled states detected by $\ew{}$. Now,  an EW $\ew{}$ is optimal if there is no other witness that is finer than $\ew{}$. One proves \cite{Lewenstein00-OEW} that $\ew{}$ is optimal iff for any $\alpha > 0$ and a positive operator $P$ an operator $\ew{} - \alpha P$ is no longer an EW. Authors of \cite{Lewenstein00-OEW} provided the following sufficient condition of optimality: for a given EW $\ew{}$ one defines
\begin{equation}\label{eq:opty}
  P_{\ew{}} = \{\, \ket{\psi \ot \phi} \in \H_A \ot \H_B\, |\, \bra{\psi \ot \phi}\ew{}\ket{\psi \ot \phi} = 0 \, \} \ .
\end{equation}
If $P_{\ew{}}$ spans $\H_A \ot \H_B$, then $\ew{}$ is optimal.

Using well known duality between bi-partite operators in $\H_A \ot \H_B$ and linear maps $\Lambda : \B(\H_A) \rightarrow\B(\H_B)$ one associates with a given EW $\ew{}$ a linear positive map by $\Lambda_{\ew{}}$ such that  $ \ew{} = (\I \ot \Lambda_{\ew{}})P^+_{A}$, where $P^+_{A}$ denotes maximally entangled state in $\H_A \ot \H_A$, and $\I$ denotes an identity map. Due to the fact that $\ew{} \ngeq 0$ the corresponding map $\Lambda_{\ew{}}$ is not completely positive (CP). 

In the present paper we provide a recurrent construction a family of positive maps  $\Psi_{N} : \m{2}^{\ot N} \rightarrow \m{2}^{\ot N}$ for $N\geq 1$. Equivalently, we define a family of EWs $\ew{N}$ in $\mathbb{C}^{2\ot N} \ot  \mathbb{C}^{2\ot N}$. Interestingly, $\Psi_1$  reproduces well known  reduction map and for $N=2$ our construction reproduces the Robertson map \cite{Robertson85-RM}. However, for $N \geq 3$ it provides brand new positive maps (equivalently EWs). Moreover, we show that for $N>1 $ these EWs are indecomposable and optimal and hence may be used to detect PPT entangled states. Finally, we show that so called structural physical approximation to $\ew{N}$ is a separable state \cite{Korbicz08-SPA}. As a byproduct we provide PPT entangled states detected by our witnesses.

\section{Recurrent construction}

In what follows we provide a recurrent construction of linear positive maps
\begin{equation*}\label{}
\Psi_{N} : \m{2}^{\ot N} \longrightarrow \,\m{2}^{\ot N},
\end{equation*}
where $\m{2}^{\ot N}$ denotes a tensor product of $N$ copies of $\m{2}$ (a space of $2\times 2$ complex matrices). Let us start with a ``vacuum'' map $\Psi_0 : \C \rightarrow \C$ defined by $\Psi_0(z) = 0$ which is evidently positive but not very interesting. Out of $\Psi_0$ we construct a family of nontrivial positive maps {\em via} the following formula
\begin{equation}\label{eq:Psi_N}
    \Psi_{N+1}\left(\begin{array}{c|c} X_{11} & X_{12} \\\hline  X_{21} &  X_{22} \end{array}\right)  = \frac{1}{2^N} \left( \begin{array}{c|c} D_{11} & -A_{N} \\ \hline -B_{N} & D_{22} \end{array} \right)
\end{equation}
with the diagonal blocks defined as
\begin{equation*}
D_{ii}= \one_{2}^{\ot N}(\tr \,X - \tr\, X_{ii})
\end{equation*}
and the off-diagonal blocks given recursively by
\begin{eqnarray*}\label{}
    A_{N} &=& X_{12} + \Psi_{N}(X_{21}) , \\
    B_{N} &=& X_{21} + \Psi_{N}(X_{12}) .
\end{eqnarray*}
In Eq. (\ref{eq:Psi_N}) one uses $\m{2}^{\ot (N+1)} = \,\m{2} \ot \m{2}^{\ot N}$ and hence we can rewrite $X= \sum_{i,j=1}^2 e_{ij} \ot X_{ij}$, with $X_{ij} \in \m{2}^{\ot N}$ and $e_{ij} = \proj{i}{j}$.  It is clear from the construction that each $\Psi_{N}$ is trace-preserving and unital, i.e. $\Psi_{N}(\one_{2}^{\ot N}) = \one_{2}^{\ot N}$.

Interestingly, one finds $\Psi_1 : \m{2} \rightarrow \m{2}$ to be
\begin{equation*}\label{}
\Psi_1\left(\begin{array}{cc} x_{11} & x_{12} \\ x_{21} & x_{22} \end{array}\right)  = \left( \begin{array}{cc} x_{22} & -x_{12} \\ -x_{21} & x_{11} \end{array} \right),
\end{equation*}
which reconstructs the reduction map in $\m{2}$, i.e.,
\begin{equation*}\label{}
    \Psi_1(X)\equiv\R(X) = \one_2\tr X - X\ .
\end{equation*}
This map is known to be positive, decomposable and optimal (even extremal) \cite{Pytel11-EW3}. Similarly one can reproduce the Robertson map:
\begin{equation*}\label{}
    \Psi_2\left( \begin{array}{c|c} X_{11} & X_{12} \\\hline  X_{21} &  X_{22} \end{array}\right)  = \frac{1}{2} \left( \begin{array}{c|c} \one_2\tr X_{22} & - A_1 \\ \hline- B_1 & \one_2 \tr X_{11} \end{array} \right)
\end{equation*}
with
\begin{eqnarray*}\label{}
    A_{1} &=& X_{12} + \R(X_{21}) , \\
    B_{1} &=& X_{21} + \R(X_{12}) ,
\end{eqnarray*}
which is known to be positive, indecomposable and extremal \cite{Pytel09-EW1}. Recently, this map has been generalized to higher dimensional bipartite systems in several ways \cite{Pytel09-EW1,Pytel10-EW2,Pytel11-EW3,Zwolak13-EW4}. In all cases these generalizations lead to families of indecomposable and optimal maps.

\section{Properties of $\Psi_N$}

In this section we analyze the basic properties of the family of maps $\Psi_N$.
We already noted that $\Psi_{N}$ is positive for $N=0,1$ and $2$ (actually, the ``vacuum'' map $\Psi_0$ is even CP). The crucial result of this paper consists in the following
\begin{thm}\label{thm:positivity}
The map $\Psi_N$ is positive for any $N$.
\end{thm}
\begin{proof} See the Appendix. 
\end{proof}
Note that for $N\ge1$ the map $\Psi_N$ is not CP. Indeed, the corresponding EW
 $\ew{N}=(\one_{N} \ot\Psi_N)P^+$  possesses exactly one negative eigenvalue
\begin{equation*}
\ew{N}\phi^{+}=-\frac{1}{2^{N}}\phi^{+},
\end{equation*}
where $\phi^{+}=\sum_{i=1}^{2^N}e_{i}\ot e_{i}$ denotes the (unnormalized) maximally entangled state. The existence of a negative eigenvalue of $\ew{N}$ proves that $\Psi_N$ is not CP and hence $\ew{N}$ is a legitimate entanglement witness.

\begin{figure}[t]
\includegraphics[width=1\columnwidth]{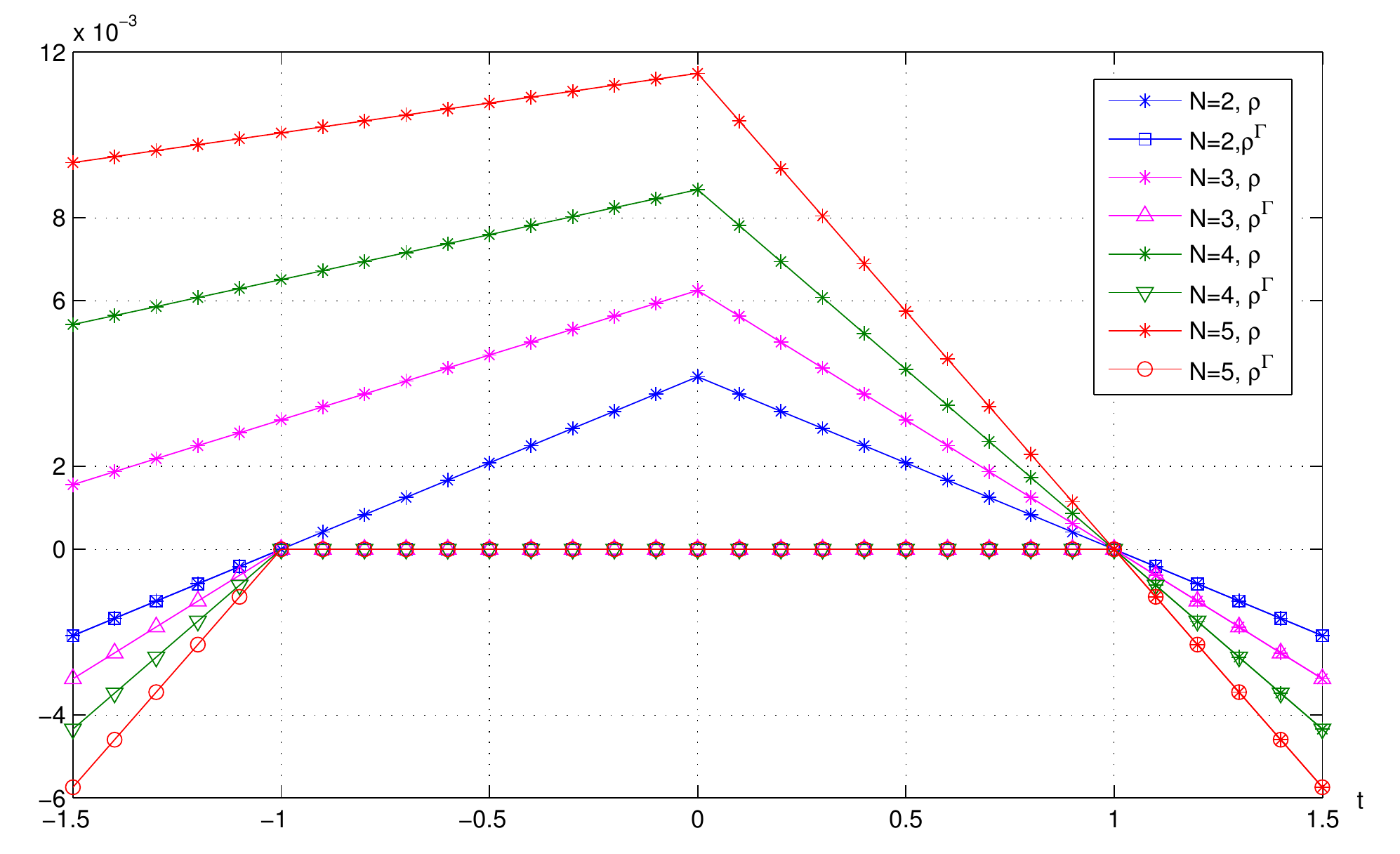}\caption{\label{fig:eig_of_rho} Smallest eigenvalues of the matrix $\rho_{t}$ defined by Eq.\,(\ref{eq:rho}) and $\rho_{t}^{\Gamma}$ as a function of the parameter $t\in[-1.5;1.5]$ for three different $N$. In the case of $N=2,4,5$ eigenvalues are scaled so that everything can be shown on one plot. It does not affect the positivity of eigenvalues.}
\end{figure}
We already noticed that $\Psi_1$, corresponding to the reduction map, is decomposable while $\Psi_2$, corresponding to the Robertson map, is indecomposable. One has the following theorem,
\begin{thm}
The map $\Psi_N$ is indecomposable for $N > 1$.
\end{thm}
\begin{proof}
To prove indecomposability of $\Psi_N$ it is enough to find a PPT state $\rho$ such that $\tr(\ew{N}\,\rho)<0$. Let us consider the following construction of a family of (unnormalized) matrices parametrized by $t\in\mathbb{R}$:
\begin{equation}\label{eq:rho}
\rho_{t}=\sum_{i,j=1}^{2^{N}}e_{ij}\ot\rho_{ij},
\end{equation}
with the $2^{N}\times2^{N}$ blocks $\rho_{ij}$  defined as follows: 
\begin{itemize}
\item $\rho_{ii}=\frac{1}{2^{N}}\one_{2^{N}}-(2^{N-1}-1)W_{ii}$\quad for $i=1,\dots,2^{N}$,
\item $\rho_{ij}=\O_{2^{N}}$, if $i\neq j$ and $i,j<2^{N-1}$ or $i,j>2^{N-1}$,
\item $\rho_{i,i+2^{N-1}}=-t\cdot W_{i,i+2^{N-1}}$,
\item $ \rho_{ij}=\frac{1}{2^{N}\cdot2^{N-1}}e_{ij}$ in the remaining cases
\end{itemize}
and $W_{ij}=\frac{1}{2^N}\Psi_N(e_{ij})$. 
Figure \ref{fig:eig_of_rho} shows how the minimal eigenvalue of a state $\rho_{t}$ and the minimal eigenvalue of the partially transposed  state $\rho_{t}^{\Gamma}$  depends on the parameter $t$. The smallest eigenvalue of $\rho_{t}^{\Gamma}$ becomes strictly negative for $t<-1$ and $t>1$.  Thus $\rho_t$ is PPT if and only if $|t|\leq 1$. This statement is true for all $N>1$.

One shows that for any $N$ the expectation value of $\ew{N}$ in the state $\rho_{t}$ is given by
\begin{equation*}\label{}
\tr(\ew{N}\rho_{t})=\frac{-4t\,(2^{N}+4)+2^{N+2}}{2^{4N}}
\end{equation*}
and hence $\rho_{t}$ is entangled for $t\in(\frac{2^{N}}{2^{N}+4},1]$.
The analysis of the few first cases is shown in Figure \ref{fig:tr(Wrho)}.
\end{proof}
As a byproduct we derive a new one-parameter class of PPT entangled states in $\mathbb{C}^{2N} \ot \mathbb{C}^{2N}$.
\begin{exmp}
One finds the following matrix representation (up to an unimportant positive constant) of $\ew{2}$
\begin{equation*}\label{}
\small\left(\begin{array}{cccc|cccc|cccc|cccc}
\cdot & \cdot & \cdot & \cdot & \cdot & \cdot & \cdot & \cdot & \cdot & \cdot & -1 & \cdot & \cdot & \cdot & \cdot & -1\\
\cdot & \cdot & \cdot & \cdot & \cdot & \cdot & \cdot & \cdot & \cdot & \cdot & \cdot & \cdot & \cdot & \cdot & \cdot & \cdot\\
\cdot & \cdot & 1 & \cdot & \cdot & \cdot & \cdot & \cdot & \cdot & \cdot & \cdot & \cdot & \cdot & 1 & \cdot & \cdot\\
\cdot & \cdot & \cdot & 1 & \cdot & \cdot & \cdot & \cdot & \cdot & -1 & \cdot & \cdot & \cdot & \cdot & \cdot & \cdot\\ \hline
\cdot & \cdot & \cdot & \cdot & \cdot & \cdot & \cdot & \cdot & \cdot & \cdot & \cdot & \cdot & \cdot & \cdot & \cdot & \cdot\\
\cdot & \cdot & \cdot & \cdot & \cdot & \cdot & \cdot & \cdot & \cdot & \cdot & -1 & \cdot & \cdot & \cdot & \cdot & -1\\
\cdot & \cdot & \cdot & \cdot & \cdot & \cdot & 1 & \cdot & \cdot & \cdot & \cdot & \cdot & -1 & \cdot & \cdot & \cdot\\
\cdot & \cdot & \cdot & \cdot & \cdot & \cdot & \cdot & 1 & 1 & \cdot & \cdot & \cdot & \cdot & \cdot & \cdot & \cdot\\ \hline
\cdot & \cdot & \cdot & \cdot & \cdot & \cdot & \cdot & 1 & 1 & \cdot & \cdot & \cdot & \cdot & \cdot & \cdot & \cdot\\
\cdot & \cdot & \cdot & -1 & \cdot & \cdot & \cdot & \cdot & \cdot & 1 & \cdot & \cdot & \cdot & \cdot & \cdot & \cdot\\
-1 & \cdot & \cdot & \cdot & \cdot & -1 & \cdot & \cdot & \cdot & \cdot & \cdot & \cdot & \cdot & \cdot & \cdot & \cdot\\
\cdot & \cdot & \cdot & \cdot & \cdot & \cdot & \cdot & \cdot & \cdot & \cdot & \cdot & \cdot & \cdot & \cdot & \cdot & \cdot\\ \hline
\cdot & \cdot & \cdot & \cdot & \cdot & \cdot & -1 & \cdot & \cdot & \cdot & \cdot & \cdot & 1 & \cdot & \cdot & \cdot\\
\cdot & \cdot & 1 & \cdot & \cdot & \cdot & \cdot & \cdot & \cdot & \cdot & \cdot & \cdot & \cdot & 1 & \cdot & \cdot\\
\cdot & \cdot & \cdot & \cdot & \cdot & \cdot & \cdot & \cdot & \cdot & \cdot & \cdot & \cdot & \cdot & \cdot & \cdot & \cdot\\
-1 & \cdot & \cdot & \cdot & \cdot & -1 & \cdot & \cdot & \cdot & \cdot & \cdot & \cdot & \cdot & \cdot & \cdot & \cdot
\end{array}\right)
\end{equation*}
and the (unnormalized) matrix $\rho_t$
\begin{equation*}\label{}
\left(\begin{array}{cccc|cccc|cccc|cccc}
2 & \cdot & \cdot & \cdot & \cdot & \cdot & \cdot & \cdot & \cdot & \cdot & t & \cdot & \cdot & \cdot & \cdot & 1\\
\cdot & 2 & \cdot & \cdot & \cdot & \cdot & \cdot & \cdot & \cdot & \cdot & \cdot & \cdot & \cdot & \cdot & \cdot & \cdot\\
\cdot & \cdot & 1 & \cdot & \cdot & \cdot & \cdot & \cdot & \cdot & \cdot & \cdot & \cdot & \cdot & \cdot & \cdot & \cdot\\
\cdot & \cdot & \cdot & 1 & \cdot & \cdot & \cdot & \cdot & \cdot & t & \cdot & \cdot & \cdot & \cdot & \cdot & \cdot\\
\hline \cdot & \cdot & \cdot & \cdot & 2 & \cdot & \cdot & \cdot & \cdot & \cdot & \cdot & \cdot & \cdot & \cdot & \cdot & \cdot\\
\cdot & \cdot & \cdot & \cdot & \cdot & 2 & \cdot & \cdot & \cdot & \cdot & 1 & \cdot & \cdot & \cdot & \cdot & t\\
\cdot & \cdot & \cdot & \cdot & \cdot & \cdot & 1 & \cdot & \cdot & \cdot & \cdot & \cdot & t & \cdot & \cdot & \cdot\\
\cdot & \cdot & \cdot & \cdot & \cdot & \cdot & \cdot & 1 & \cdot & \cdot & \cdot & \cdot & \cdot & \cdot & \cdot & \cdot\\
\hline \cdot & \cdot & \cdot & \cdot & \cdot & \cdot & \cdot & \cdot & 1 & \cdot & \cdot & \cdot & \cdot & \cdot & \cdot & \cdot\\
\cdot & \cdot & \cdot & t & \cdot & \cdot & \cdot & \cdot & \cdot & 1 & \cdot & \cdot & \cdot & \cdot & \cdot & \cdot\\
t & \cdot & \cdot & \cdot & \cdot & 1 & \cdot & \cdot & \cdot & \cdot & 2 & \cdot & \cdot & \cdot & \cdot & \cdot\\
\cdot & \cdot & \cdot & \cdot & \cdot & \cdot & \cdot & \cdot & \cdot & \cdot & \cdot & 2 & \cdot & \cdot & \cdot & \cdot\\
\hline \cdot & \cdot & \cdot & \cdot & \cdot & \cdot & t & \cdot & \cdot & \cdot & \cdot & \cdot & 1 & \cdot & \cdot & \cdot\\
\cdot & \cdot & \cdot & \cdot & \cdot & \cdot & \cdot & \cdot & \cdot & \cdot & \cdot & \cdot & \cdot & 1 & \cdot & \cdot\\
\cdot & \cdot & \cdot & \cdot & \cdot & \cdot & \cdot & \cdot & \cdot & \cdot & \cdot & \cdot & \cdot & \cdot & 2 & \cdot\\
1 & \cdot & \cdot & \cdot & \cdot & t & \cdot & \cdot & \cdot & \cdot & \cdot & \cdot & \cdot & \cdot & \cdot & 2
\end{array}\right)
\end{equation*}
where $2^N \times 2^N$ blocks are separated by horizontal and vertical lines. Moreover, to make the picture more transparent we denote zeros by dots.
One finds
\begin{equation*}
\tr(\ew{2}\rho_{t}) = (4-8t)/8
\end{equation*} 
which shows that $\rho_{t}$ is entangled for $t > 1/2$.
\end{exmp}
\begin{figure}[t]
\includegraphics[width=1\columnwidth]{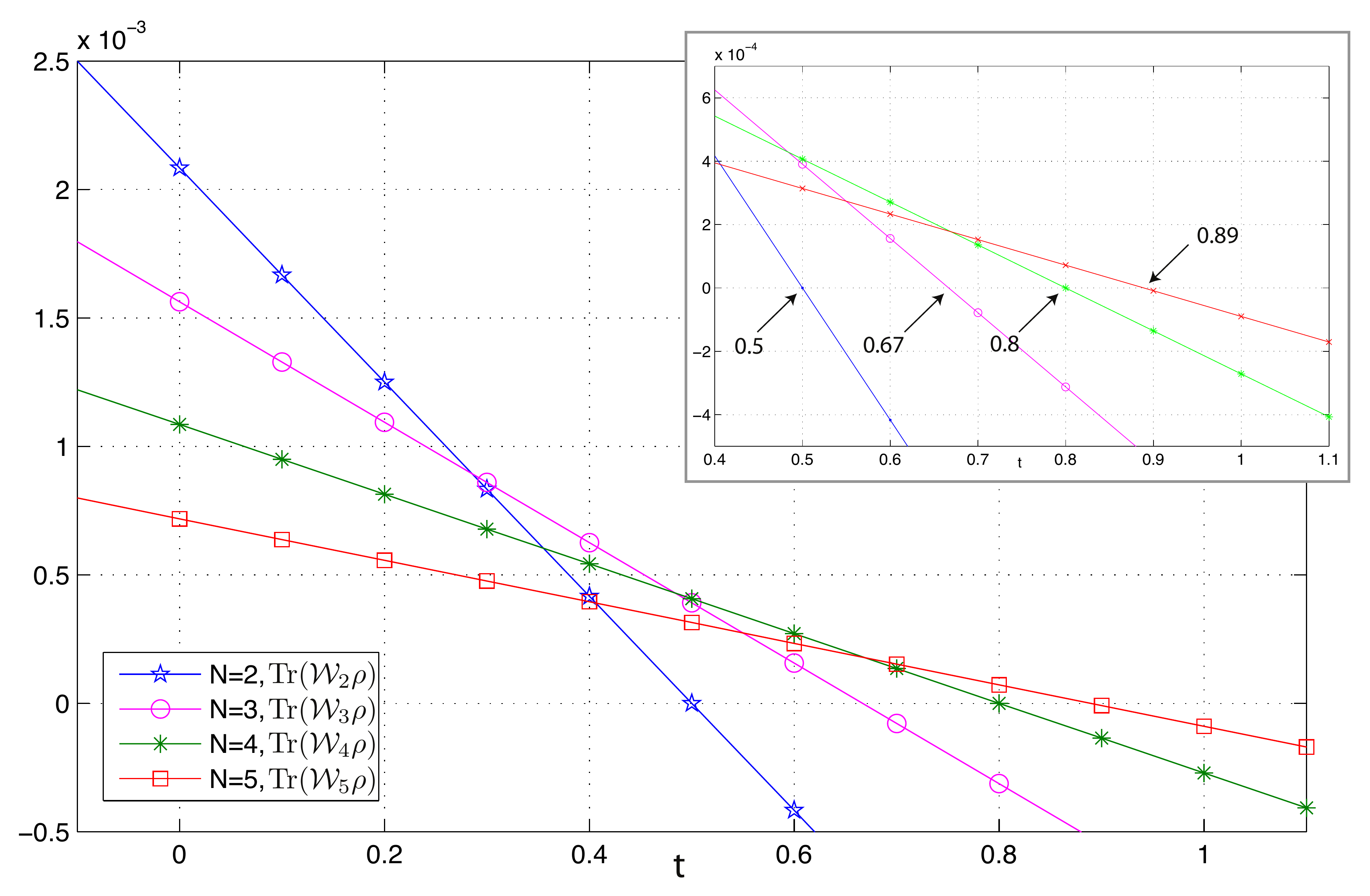}\caption{\label{fig:tr(Wrho)} The expectation value of $\ew{N}$ in a state
$\rho_{t}$ for three different values of $N$.}
\end{figure}

One can observe that with the increase of the number of qubits, the range of $t$ gets smaller. The decreasing range of $t$ can be intuitively ascribed to the fact that having more qubits in our system leads to spreading out the same ``amount'' of entanglement between more particles. As a consequence, our witness $\ew{N}$ might become not strong enough to detect it. In order to fully understand how the entanglement is being distributed in states $\rho_t$ and $\rho_t^{\Gamma}$ further and more detailed analysis is necessary.

Following Ref. \cite{Lewenstein00-OEW} to prove that $\Psi_{N}$ are optimal for $N>1$ it is enough to find for each $N$  a set of linearly independent product vectors $\psi_{i}\ot\phi_{i}\in\C^{2\ot N} \ot \C^{2\ot N}$ satisfying Eq. (\ref{eq:opty}). Let us consider a set of vectors introduced in Ref. \cite{Pytel09-EW1}:
\begin{equation*}
\mathcal{G}_{\ew{}}:=\{\psi_{\alpha}\ot\psi_{\alpha}^{*},\alpha=1,\dots,2^{2N}\}
\end{equation*}
with $\psi_{\alpha}\in\{e_{l}\,,f_{mn}\,,g_{mn}\}$, where $\{e_{i}\}$ stands for an orthonormal basis and
\begin{eqnarray*}
f_{mn} & = & e_{m}+e_{n},\\
g_{mn} & = & e_{m}+\imath e_{n} ,
\end{eqnarray*}
for $1\leq m<n\leq N$. Direct calculations show that elements of $\mathcal{G}_{\ew{}}$
are linearly independent and that
\begin{equation*}
\forall_{\alpha=1,\dots,N}\quad\inner{\psi_{\alpha}\ot\psi_{\alpha}^{*}|\ew{N}}{\psi_{\alpha}\ot\psi_{\alpha}^{*}}=0,
\end{equation*}
which is sufficient to prove the the following theorem:
\begin{thm}
For all $N\ge1$, $\Psi_N$ defines a class of optimal maps.
\end{thm}

Positive, but not completely positive maps, unlike entanglement witnesses, cannot be directly implemented in the laboratory. One way to tackle this problem is to approximate the positive map by a completely positive one which may serve as a quantum operation. Given a positive map $\Lambda : \B(\H) \rightarrow \B(\H)$ one defines a family of maps
\begin{equation*}\label{}
   \widetilde{\Lambda}(p)= p\,\I+(1-p)\Lambda\ .
\end{equation*}
Let $p_*$ be the smallest $p$ such that $\widetilde{\Lambda}(p_*)$ is completely positive. One calls $\widetilde{\Lambda}(p_*)$ the structural physical approximation (SPA) of $\Lambda$. It was conjectured \cite{Korbicz08-SPA,Augusiak11-SPA} that structural physical approximation to an optimal positive map defines an entanglement breaking map (a completely positive map $\mathcal{E}$ is  entanglement breaking if $(\I \ot \mathcal{E})\rho$ is separable for an arbitrary state $\rho$, see Ref. \cite{Horodecki03-EBC}). In the language of EWs SPA conjecture states that if $\ew{}$ is an optimal EW, then the corresponding SPA
\begin{equation*}\label{}
  \ew{}(p_*) = \frac{p_*}{d_A d_B} \one_A \ot \one_B + (1-p_*)\ew{}\ ,
\end{equation*}
defines a separable state. Recently SPA conjecture has been disproved for indecomposable EWs in \cite{Kye12-SPA} and for decomposable ones in \cite{Chruscinski13-DSC} (see also recent papers \cite{Augusiak13-COE,Wang13-SPA}). Interestingly, the SPA for $\Psi_N$ provides  EB map.
To show this let us recall the following result from Ref. \cite{Pytel11-EW3}
\begin{cor}
If $\Lambda: \m{n} \rightarrow \m{n}$ is a unital map, and the smallest eigenvalue of the corresponding entanglement witness $W$ satisfies $\xi_{min}\leq -\frac{1}{n}$, then the SPA to $W$ defines a separable state.
\end{cor}
Since for any $N\ge1$ an entanglement witness $\ew{N}$ corresponding
to $\Psi_{N}$ posses only one negative eigenvalue $\xi=-\frac{1}{2^N}$,
thus the SPA to $\Psi_{N}$ indeed defines an entanglement breaking
channel.

\section{Conclusions}

We provided a new class of linear positive, but not completely positive, maps in $\m{2}^{\ot N}$. These maps are  indecomposable and optimal, and their structural physical approximation gives rise to an entanglement breaking channel. Equivalently, our construction provide new entanglement witnesses for bi-partite systems where each subsystem lives in the $N$ qubit Hilbert space.


\begin{acknowledgments}
DC was partially supported by the National Science Centre project
DEC-2011/03/B/ST2/00136.
\end{acknowledgments}

%

\section*{Appendix: proof of Theorem 1}
\begin{proof}
We prove the theorem by induction. We already known that it holds for $N=1$ and $N=2$. Now, assuming that it is true for $\Psi_N$ we prove it for $\Psi_{N+1}$. We shall use the fact that $\Psi_N$ is contractive, i.e.
\begin{equation}\label{eq:contr}
  \|\Psi_N(X)\| \leq \|X\|  \ ,
\end{equation}
where $\|X\|$ denotes an operator norm of $X$, i.e., the maximal eigenvalue of $|X|= \sqrt{XX^\dg }$. Recall that any unital map is positive iff it is contractive in the operator norm \cite{Bhatia-PDM}. To show that $\Psi_{N+1}$ defines a positive map  it is enough to show that it maps any rank-1 projector into a positive element. Let us consider $P=\proj{\psi}{\psi}$ with $\psi$ being an arbitrary vector in $\C^{2^{N+1}}$. Since $\C^{2^{N+1}}=\C^{2^{N}}\oplus\C^{2^{N}}$ one can rewrite $\psi=\bigoplus_{i=1}^2\sqrt{\alpha_{i}}\psi_{i}$, with $\psi_{1},\psi_{2}\in\C^{2^{N}}$ and $\alpha_{1}+\alpha_{2}=1$. Without loosing generality one can assume $\inner{\psi_{i}}{\psi_{i}}=1$ and hence
\begin{equation*}
\Psi_{N+1}(P)=\frac{1}{2^{N}}\left(\begin{array}{c|c}
\one_{2^{N}}\alpha_{2} & -\sqrt{\alpha_{1}\alpha_{2}}\,A_{N}\\ \hline
-\sqrt{\alpha_{1}\alpha_{2}}\,A_{N}^{\dg} & \one_{2^N}\alpha_{1}
\end{array}\right)\ ,
\end{equation*}
with $A_{N}=\proj{\psi_{1}}{\psi_{2}}+\Psi_{N}(\proj{\psi_{2}}{\psi_{1}}).$
It is clear that $\Psi_{N+1}(P) \geq 0$  iff
\begin{equation}\label{ineq:AA^dg}
A_{N}A_{N}^{\dg}\le\one_{2^N} .
\end{equation}

\begin{lem}\label{lem:prop}The map $\Psi_N$ satisfies
\begin{equation}\label{eq:prop}
  \Psi_N(\proj{x}{y})\ket{x} = 0\ , \ \ \ \bra{y}\Psi_N(\proj{x}{y}) = 0 \ ,
\end{equation}
for any vectors $\ket{x},\ket{y} \in \C^{2^N}$.
\end{lem}
\begin{proof} We prove this by induction. For $N=1$ one immediately verifies (\ref{eq:prop}). Now, assuming  that (\ref{eq:prop}) holds for $\Psi_N$ we prove it for $\Psi_{N+1}$. Using
\begin{equation*}\label{}
  \ket{x} = \ket{x_1 \oplus x_2} \ , \ \ket{y} = \ket{y_1\oplus y_2} \ ,
\end{equation*}
one finds for $ 2^N \Psi_{N+1}(\proj{x}{y})$:
\begin{equation*}\label{}
  \left( \begin{array}{c|c} \inner{y_2}{x_2} \one_{2^N} & - \proj{x_1}{y_2} - \Psi_N(\proj{x_2}{y_1}) \\ \hline
  - \proj{x_2}{y_1} - \Psi_N(\proj{x_1}{y_2}) & \inner{y_1}{x_1} \one_{2^N} \end{array} \right)  \ ,
\end{equation*}
and hence
\begin{equation*}
  \Psi_{N+1}(\proj{x}{y})\ket{x} \equiv \Psi_{N+1}(\proj{x}{y}) \left( \begin{array}{c} \ket{x_1} \\ \hline \ket{x_2} \end{array} \right) = 0 \
\end{equation*}
where we have used $\Psi_N(\proj{x_2}{y_1}) \ket{x_2} = 0$. Similarly $\bra{y}\Psi_N(\proj{x}{y}) = 0$.
\end{proof}
Now, using the Lemma \ref{lem:prop} one arrives at
\begin{equation*}\label{}
A_N A_N^\dg = \proj{\psi_1}{\psi_1} + Q_N \ ,
\end{equation*}
where $Q_N = \Psi_N(\proj{\psi_2}{\psi_1}) \Psi_N(\proj{\psi_1}{\psi_2})$.
Note that $Q_N$ is supported on the subspace orthogonal to $\ket{\psi_1}$ and hence the set of eigenvalues of $A_N A_N^\dg$ consists of eigenvalues of $Q_N$ and 1. Now, using contractivity (\ref{eq:contr}) one obtains
\begin{equation*}\label{}
\| \Psi_N(\proj{\psi_1}{\psi_2}) \| \leq \|\proj{\psi_1}{\psi_2}\| \leq 1 \ ,
\end{equation*}
which shows that the maximal eigenvalue of $Q_N$ is not greater than 1. This finally proves (\ref{ineq:AA^dg}).
\end{proof}

\end{document}